\documentclass[11pt]{article}
\usepackage{amssymb, amsmath, amsthm}
\usepackage[round]{natbib}
\usepackage{amsfonts,mathtools,mathrsfs,amssymb,amsthm,color}
\usepackage[colorlinks]{hyperref}
\usepackage{graphicx}
\usepackage{bm}

\newtheorem{theorem}{Theorem}[section]
\newtheorem{lemma}[theorem]{Lemma}
\newtheorem{corollary}[theorem]{Corollary}

\newtheorem{remark}{Remark}

\newcommand{\g}{\mathsf{g}}

\newcommand{\R}{\mathbb{R}} 
\newcommand{\N}{\mathcal{N}}

\newcommand{\E}{\mathbb{E}}

\renewcommand{\O}{\mathcal{O}}

\renewcommand{\part}[2]{\frac{\partial #1}{\partial #2}}

\DeclareMathOperator{\Tr}{Tr}

\newcommand{\Var}{\mathrm{Var}}

\newcommand{\HS}{\mathrm{HS}}

\newcommand{\MLA}{\mathrm{MLA}}

\newcommand{\Alg}{\mathsf{Alg}}

\newcommand{\X}{\mathcal{X}}
\newcommand{\Y}{\mathcal{Y}}
\newcommand{\grad}{\mathrm{grad}}

\title{The Mirror Langevin Algorithm Converges with Vanishing Bias}

\author{
Ruilin Li\thanks{Georgia Institute of Technology \& Hudson River Trading. Email: \texttt{liruilin1993@gmail.com} }
\and
Molei Tao\thanks{Georgia Institute of Technology, School of Mathematics. Email: \texttt{mtao@gatech.edu} }
\and
Santosh S.\ Vempala\thanks{Georgia Institute of Technology, College of Computing. Email: \texttt{vempala@gatech.edu} }
\and 
Andre Wibisono\thanks{Yale University, Department of Computer Science. Email: \texttt{andre.wibisono@yale.edu}}}

\begin{document}

\maketitle

\begin{abstract}
The technique of modifying the geometry of a problem from Euclidean to Hessian metric has proved to be quite effective in optimization, and has been the subject of study for sampling. The Mirror Langevin Diffusion (MLD) is a sampling analogue of mirror flow in continuous time, and it has nice convergence properties under log-Sobolev or Poincare inequalities relative to the Hessian metric, as shown by Chewi et al.\ (2020). In discrete time, a simple discretization of MLD is the Mirror Langevin Algorithm (MLA) studied by Zhang et al.\ (2020), who showed a biased convergence bound with a non-vanishing bias term (does not go to zero as step size goes to zero). This raised the question of whether we need a better analysis or a better discretization to achieve a vanishing bias. 
Here we study the Mirror Langevin Algorithm and show it indeed has a vanishing bias. We apply mean-square analysis based on Li et al.\ (2019) and Li et al.\ (2021) to show the mixing time bound for MLA under the modified self-concordance condition introduced by Zhang et al.\ (2020).
\end{abstract}

\section{Introduction}

Suppose we wish to sample from a probability distribution $\nu(x) \propto e^{-f(x)}$ supported on a convex set $\X \subseteq \R^d$
where $f \colon \X \to \R$ is differentiable.
A popular algorithm is the Unadjusted Langevin Algorithm (ULA), which is a basic discretization of the Langevin Dynamics in continuous time:
\begin{align*}
    dX_t = -\nabla f(X_t)\,dt + \sqrt{2}\,dW_t.
\end{align*}
Langevin Dynamics has an optimization interpretation as the gradient flow for minimizing relative entropy (KL divergence) with respect to $\nu$ using the Wasserstein metric $W_2$ in the space of probability distributions on $\R^d$, starting from the seminal work of~\cite{JKO98}; see also~\cite{W18}.
In continuous time, Langevin Dynamics has convergence guarantees in various distances, including $W_2$ distance, KL divergence, or $\chi^2$-divergence, under various conditions, such as strong log-concavity, or functional inequalities such as the Log-Sobolev Inequality (LSI) or Poincar\'e inequality.
In discrete time, ULA is a biased discretization of the Langevin Dynamics, and it has an asymptotic bias which scales with the step size.
In particular, by setting a small enough step size, we can obtain a mixing time bound of ULA which has inverse polynomial dependence on the error threshold; see for example~\cite{D17, durmus2017nonasymptotic, durmus2019high, DK19, VW19, li2019stochastic};~\cite{LZT21}.

In many settings, the problem of interest is non-smooth or constrained (e.g., the $L_1$ ball or a general polytope), and the basic Langevin algorithm does not apply. In optimization, this is handled effectively (and elegantly) by interior-point methods, which use a convex ``barrier" function to define a non-Euclidean metric. The resulting metric is given locally by the Hessian of the barrier function. This method results is a convergence bound that scales as $\sqrt{d}$ for linear and convex optimization.

It is natural to wonder whether such a modification of the geometry could be useful for sampling. Early evidence of this is the Dikin walk, which replaces the ball walk (a discrete-time implementation of constrained Brownian motion) by using an ellipsoid at each step, defined by the Hessian of the logarithmic barrier function. This walk was shown to converge in $\tilde{O}(md)$ steps for uniformly sampling a polytope with $m$ facets in $d$ dimension~\citep{KN12}. It was recently refined using a weighted barrier function to improve the convergence time to $\tilde{O}(d^2)$~\citep{LLV20}. 

A related approach that also originated in optimization is {\em mirror descent}, which uses a mirror map (the gradient of the barrier function) to change the geometry favorably, and in the context of sampling, can be seen as a generalization of the Langevin algorithm by changing the metric. 
For constrained sampling, the Mirror Langevin Dynamics was introduced by~\cite{ZPFP20} using a mirror map to constrain the domain; see also~\cite{Hsieh18} for a related earlier approach.
Mirror Langevin Dynamics is the Langevin dynamics for sampling from $\nu$ using the Hessian metric generated by the mirror map.
In continuous time, Mirror Langevin Dynamics has nice convergence guarantees under an analogous notion of mirror Poincare inequality relative to the Hessian metric, as shown by~\cite{ChewiEtAl20}; see also Appendix~\ref{App:CtsTime} for a review.
In discrete time, the Mirror Langevin Algorithm (MLA) is a simple discretization of MLD proposed by~\cite{ZPFP20}, who showed a biased convergence analysis but with a non-vanishing bias (does not go to $0$ with step size, but remains a constant).
This raised the question of whether we need a better analysis of MLA or a better discretization of MLD to achieve a vanishing bias.
This also led to an alternative discretization proposed by~\cite{AC20} which achieves a vanishing bias, but requires an exact simulation of the Brownian motion with changing covariance.

In this paper, we study the Mirror Langevin Algorithm and show that it indeed has a vanishing bias. 
The tool we will use is the mean-square analysis framework, proposed by~\cite{li2019stochastic} and then refined by~\cite{LZT21}; the latter version will be used. It will help establish a biased convergence analysis of MLA under relative smoothness, strong convexity and modified self-concordance; these are a subset of the conditions assumed by~\cite{ZPFP20}.
We show that the bias of MLA with step size $h$ scales as $\sqrt{h}$;
this leads to a $\tilde O(d/\epsilon^2)$ mixing time bound for MLA (see Theorem~\ref{Thm:Main} and Corollary~\ref{Cor:Main}).

\section{Algorithm and Problem Set-Up}
\label{Sec:SetUp}

\subsection{Problem set-up}

Suppose we want to sample from a probability distribution $\nu$ supported on a convex set $\X \subseteq \R^d$.
We assume $\nu$ is absolutely continuous with respect to the Lebesgue measure on $\R^d$ and has density $\nu(x) \propto e^{-f(x)}$ for some differentiable $f \colon \X \to \R$.

Let $\phi \colon \X \to \R$ be a twice-differentiable strictly convex function which is of {\em Legendre type}~\citep{Roc70}.
This implies $\nabla \phi(\X) = \R^d$, and in particular the gradient map $\nabla \phi \colon \X \to \R^d$ is bijective.
We also have $\nabla^2 \phi(x) \succ 0$ for all $x \in \X$. 
Moreover, we require that $\|\nabla \phi(x)\| \to \infty$ and $\nabla^2 \phi(x) \to \infty$ as $x$ approaches the boundary of $\X$.
Using the Hessian metric $\nabla^2 \phi$ on $\X$ will prevent the iterates from leaving the domain $\X$.
We call $\nabla \phi \colon \X \to \R^d$ the {\em mirror map} and $\Y = \nabla \phi(\X) = \R^d$ the {\em dual space}.

Let $\phi^\ast \colon \R^d \to \R$ be the {\em dual function} of $\phi$, defined by $\phi^\ast(y) = \sup_{x \in \X} \langle x,y \rangle - \phi(x)$.
Recall $\nabla \phi^\ast(y) = \arg \max_{x \in \X} \langle x,y \rangle - \phi(x)$, and we have $\nabla \phi^\ast = (\nabla \phi)^{-1}$, so $\nabla \phi(\nabla \phi^*(y)) = y$ for all $y \in \R^d$.
Furthermore, $\nabla^2 \phi(x) = \nabla^2 \phi^*(\nabla \phi(x))^{-1}$ for all $x \in \X$.

For a vector $v \in \R^d$, let $\|v\| = \sqrt{\langle v,v \rangle}$ be the $\ell_2$-norm.
For a matrix $A \in \R^{d \times d}$, let $\|A\|_{\HS} = \sqrt{\Tr(AA^\top)}$ be the Hilbert-Schmidt norm.

\subsection{Mirror Langevin Algorithm}

In this paper we study the {\bf Mirror Langevin Algorithm}:
\begin{align}\label{Eq:MLA}
x_{k+1} = \nabla \phi^*\left( \nabla \phi(x_k) - h \nabla f(x_k) + \sqrt{2h} \sqrt{\nabla^2 \phi(x_k)} \, z_k \right)
\end{align}
where $h > 0$ is step size and $z_k \sim \N(0,I)$ is an independent Gaussian random variable in $\R^d$.
Here $\sqrt{\nabla^2 \phi(x)}$ is a square-root of $\nabla^2 \phi(x)$, i.e.\ any matrix $C(x) \in \R^{d \times d}$ satisfying $C(x)C(x)^\top = \nabla^2 \phi(x)$.
This algorithm can be seen as a sampling version of the mirror descent algorithm from optimization, since we can write the update of MLA in the following form which resembles mirror descent:
\begin{align*}
x_{k+1} = \arg\min_{x \in \X} \left\{ \langle h \nabla f(x_k) - \sqrt{2h} \sqrt{\nabla^2 \phi(x_k)} \, z_k, \, x-x_k \rangle + D_\phi(x,x_k) \right\}
\end{align*}
where $D_\phi(x,x') = \phi(x) - \phi(x') - \langle \nabla \phi(x'), x-x' \rangle$ is the Bregman divergence of $\phi$.
In particular, in the {\em Euclidean case}, i.e.\ when $\X = \R^d$ and $\phi(x) = \frac{1}{2} \|x\|^2$, MLA recovers the usual Unadjusted Langevin Algorithm (ULA).

MLA can be seen as a coordinate-transformed Euler-Maruyama discretization of the {\bf Mirror Langevin Dynamics} in continuous time, given by
\[ \begin{cases}
    Y_t  &= \nabla\phi(X_t) \\
    dY_t &= -\nabla f(X_t) dt + \sqrt{2}\sqrt{\nabla^2\phi(X_t)}dW_t
\end{cases}. \]
See Section~\ref{sec:MLAinDualSpace} for a reformulation purely in the dual space, and Appendix~\ref{App:CtsTime} for more details on the continuous-time dynamics.
\cite{ZPFP20} studied MLA as a simple discretization of the Mirror Langevin Dynamics, and
showed that under certain assumptions, the iterates of MLA converge to a Wasserstein ball around the target with some radius which depends on the modified self-concordance parameter of $\phi$ (see Section~\ref{Sec:DiscResult} for more detail).
In the Euclidean case (when $\phi$ is quadratic) this radius is $0$, so MLA converges to $\nu$ with sufficiently small step size, recovering the typical bound for ULA.
However, for general $\phi$, the radius is positive.
Therefore, the result of~\cite{ZPFP20} only guarantees MLA enters a ball around the target, but it may not converge to the target even when the step size goes to $0$.
They further conjectured the bias is unavoidable. 
This raises an interesting question of whether the non-vanishing bias of MLA is indeed unavoidable because we are discretizing a diffusion process with changing covariance, or whether there is a better analysis of MLA with vanishing bias.
Here we show that indeed MLA has a vanishing bias, by applying the mean-square analysis framework of~\cite{li2019stochastic} and~\cite{LZT21}.

\subsection{Mirror Langevin Algorithm in the Dual Space}
\label{sec:MLAinDualSpace}
Let us work in the dual space $\Y = \nabla \phi(\X) = \R^d$ via the mirror map $\nabla \phi \colon \X \to \R^d$.
Given $x \in \X$, we define the {\em dual variable}
\begin{align}
y = \nabla \phi(x) \in \R^d
\end{align}
and its inverse is given by $x = \nabla \phi^*(y)$.
The target distribution $\tilde \nu$ on the dual space is the pushforward of the original target $\nu \propto e^{-f}$ under the mirror map:
$\tilde \nu = (\nabla \phi)_\# \nu$.
If we write the density as $\tilde \nu(y) \propto e^{-\tilde f(y)}$, then we have
$\tilde f(y) =  f(\nabla \phi^*(y)) - \log \det \nabla^2 \phi^*(y)$.
Moreover, the Hessian metric $\nabla^2 \phi(x)$ on $\X$ corresponds to the Hessian metric $\nabla^2 \phi^*(y)$ on $\R^d$ generated by the dual function $\phi^*$; that is, $\nabla^2 \phi^*$ on $\R^d$ is the pullback metric of $\nabla^2 \phi$ on $\X$ under the inverse mirror map $\nabla \phi^* \colon \R^d \to \X$.
Therefore, the metric space $(\X, \nabla^2 \phi)$ is isometric to $(\R^d, \nabla^2 \phi^*)$.

If $x_k \in \X$ follows the Mirror Langevin Algorithm~\eqref{Eq:MLA}, then $y_k = \nabla \phi(x_k) \in \R^d$ follows the Mirror Langevin Algorithm in the dual space:
\begin{align}\label{Eq:MLA2}
y_{k+1} = y_k - h \nabla f(\nabla \phi^*(y_k)) + \sqrt{2h} \sqrt{\nabla^2 \phi^*(y_k)^{-1}} \, z_k.
\end{align}
MLA in the dual space~\eqref{Eq:MLA2} can be seen as a discretization of the mirror Langevin dynamics to sample from $\tilde \nu \propto e^{-\tilde f}$ with the Hessian metric $\nabla^2 \phi^*$ on $\R^d$.

Let us define $g \colon \R^d \to \R^d$ and $A \colon \R^d \to \R^{d \times d}$ by
\begin{align}\label{Eq:g}
g(y) &= \nabla f(\nabla \phi^*(y)) \\
A(y) &= \sqrt{\nabla^2 \phi^*(y)^{-1}}. \label{Eq:A}
\end{align}
Note here $A(y)$ is {\em any} square-root of $\nabla^2 \phi^*(y)^{-1}$.
Then we can write MLA in the dual space as
\begin{align}\label{Eq:MLA3}
y_{k+1} = y_k - h g(y_k) + \sqrt{2h} A(y_k) \, z_k.
\end{align}
As $h \to 0$, MLA converges to the {\bf Mirror Langevin Dynamics}, which is a continuous-time stochastic process $Y_t \in \R^d$ following the stochastic differential equation:
\begin{align*}
    dY_t = -g(Y_t) dt + \sqrt{2} A(Y_t) dW_t
\end{align*}
where $W_t$ is the standard Brownian motion in $\R^d$; see Section~\ref{Sec:MLD} for more properties.

\subsection{Wasserstein distance in dual space}

Along MLA in the dual space~\eqref{Eq:MLA2}, let $\tilde \rho_k$ denote the distribution of the random variable $y_k \in \R^d$.
We will show a convergence analysis of MLA in the dual space in terms of the Euclidean Wasserstein distance $W_2$ between $\tilde \rho_k$ and $\tilde \nu$ on $\R^d$:
\begin{align}
W_{2}(\tilde \rho, \tilde \nu)^2 = \inf_{y \sim \tilde \rho, y^* \sim \tilde \nu} \E[\|y - y^*\|^2].
\end{align}
Note that this distance does not use the Hessian metric $\nabla^2 \phi^*$ on $\R^d$.
In the original space $\X$, this gives a modified $W_2$ distance under the mirror map:
\begin{align}
W_{2,\phi}(\rho,\nu)^2 = \inf_{x \sim \rho, x' \sim \nu} \E[\|\nabla \phi(x) - \nabla \phi(x')\|^2].
\end{align}
That is, if $\tilde \rho = (\nabla \phi)_\# \rho$ and $\tilde \nu = (\nabla \phi)_\# \nu$, then $W_{2,\phi}(\rho,\nu) = W_2(\tilde \rho, \tilde \nu)$.
This is the same modified Wasserstein distance that is used in~\cite{ZPFP20}.
This corresponds to using the {\em squared} Hessian metric $(\nabla^2 \phi(x))^2$ on $\X$, which is isometric to the Euclidean metric $I$ on $\R^d$ (rather than the Hessian metric $\nabla^2 \phi(x)$ on $\X$, which is isometric to the Hessian metric $\nabla^2 \phi^*(y)$ on $\R^d$, and which is used in the continuous-time analysis in~\cite{ChewiEtAl20}).

\section{Main Result: Mixing Time Bound for MLA}

We present our main result on the mixing time bound of MLA.
We need the following assumptions. 

\begin{enumerate}
    \item[{\bf (A1)}] $\phi$ satisfies the {\bf modified self-concordance} property with parameter $\alpha > 0$, which means:
    \begin{align}
        \|\sqrt{\nabla^2 \phi(x')} - \sqrt{\nabla^2 \phi(x)}\|_{\HS} \le \sqrt{\alpha} \|\nabla \phi(x')-\nabla \phi(x)\|_2 ~~~~~ \forall \, x',x \in \X.
    \end{align}
    Equivalently, $A(y) = \sqrt{\nabla^2 \phi^*(y)^{-1}}$ is $\sqrt{\alpha}$-Lipschitz in the Hilbert-Schmidt norm:
    \begin{align}
        \|A(y') - A(y)\|_{\HS} \le \sqrt{\alpha} \|y'-y\|_2 ~~~~~ \forall \, y',y \in \R^d.
    \end{align}

    \item[{\bf (A2)}] $f$ is {\bf $M$-smooth} with respect to $\phi$ for some $0 < M < \infty$, which means:
     \begin{align}
        \|\nabla f(x')-\nabla f(x)\|_2 \le M \|\nabla \phi(x')-\nabla \phi(x)\|_2  ~~~~~~ \forall \, x',x \in \X.
    \end{align}
   Equivalently, $g(y) = \nabla f(\nabla \phi^*(y))$ is $M$-Lipschitz:
    \begin{align}
        \|g(y') - g(y)\|_2 \le M \|y'-y\|_2  ~~~~~~ \forall \, y',y \in \Y.
    \end{align}

    \item[{\bf (A3)}] $f$ is {\bf $m$-strongly convex} with respect to $\phi$ for some $0 < m \le M$, which means:
    \begin{align}
        \langle \nabla f(x') - \nabla f(x), \nabla \phi(x')-\nabla \phi(x) \rangle \ge m \|\nabla \phi(x')-\nabla \phi(x)\|^2_2 ~~~~~ \forall \, x',x \in \X.
    \end{align}
   Equivalently, $g(y) = \nabla f(\nabla \phi^*(y))$ is $m$-monotone:
    \begin{align}
        \langle g(y') - g(y), y'-y \rangle \ge m \|y'-y\|^2_2 ~~~~~ \forall \, y',y \in \R^d.
    \end{align}

\end{enumerate}

These are a subset of the assumptions in~\cite{ZPFP20}. 
In particular, we do not assume a bound on the commutator of $\nabla^2 f$ and $\nabla^2 \phi$.
Our main result is the following.

\begin{theorem}\label{Thm:Main}
Assume {\bf (A1)}, {\bf (A2)}, {\bf (A3)}, and assume $\alpha < m$.
There is a maximum step size $h_{\max} = \O\left(\frac{(m-\alpha)^2}{M^2(1+4\alpha)^2}\right)$ and constant $C_{\MLA} = \O\left( \frac{M (1+4\alpha) \sqrt{d} }{m-\alpha} \right)$, such that if we run MLA~\eqref{Eq:MLA} with $0 < h \le h_{\max}$ from any $x_0 \sim \rho_0$, then the iterates $x_k \sim \rho_k$ satisfy:
\begin{align}
    W_{2,\phi}(\rho_k, \nu) \le \sqrt{2} e^{-(m-\alpha) h k} W_{2,\phi}(\rho_0, \nu) + C_{\MLA} \sqrt{2h}.
\end{align}
Equivalently, if we run MLA in the dual space~\eqref{Eq:MLA2} with $0 < h \le h_1$ from any $y_0 \sim \tilde \rho_0$, then the iterates $y_k \sim \tilde \rho_k$ satisfy:
\begin{align}
    W_2(\tilde \rho_k, \tilde \nu) \le \sqrt{2} e^{-(m-\alpha) h k} W_2(\tilde \rho_0, \tilde \nu) + C_{\MLA} \sqrt{2h}.
\end{align}
\end{theorem}

See Section~\ref{Sec:ProofThmMain} for the proof of Theorem~\ref{Thm:Main} and explicit forms of the constant $C_{\MLA}$ and maximum step size $h_{\max}$.
This result shows MLA has a bias that is vanishing with step size, and thus we can reach an arbitrary accuracy by using a small enough step size.
In particular, this improves on the analysis in~\cite{ZPFP20}, which has a non-vanishing bias and under stronger assumptions.
By choosing a small step size, we obtain the following mixing time bound for MLA.

\begin{corollary}\label{Cor:Main}
For any (small) error threshold $\epsilon > 0$, to reach $W_2(\tilde \rho_k, \tilde \nu) \le \epsilon$, it suffices to run MLA in the dual space~\eqref{Eq:MLA2} with step size $h = \frac{\epsilon^2}{4C_{\MLA}^2}$ for $k = \tau_{W_2}(\epsilon)$ iterations where
\begin{align*}
    \tau_{W_2}(\epsilon) 
      &\le \frac{1}{(m-\alpha) h} \log \frac{2\sqrt{2} W_2(\tilde \rho_0, \tilde \nu)}{\epsilon}
      = \tilde O\left(\frac{C_{\MLA}^2}{(m-\alpha) \epsilon^2}\right)
      = \tilde O\left(\frac{M^2 (1+4\alpha)^2 d}{(m-\alpha)^3 \epsilon^2} \right).
\end{align*}
\end{corollary}

\subsection{Discussion of result}
\label{Sec:DiscResult}

Theorem~\ref{Thm:Main} shows that MLA has a biased convergence guarantee where the bias scales as $O(\sqrt{d h})$ where $d$ is dimension and $h$ is step size (assuming $m, M, \alpha$ are independent of $d$ for now).
This leads to a mixing time bound of $\tilde O(d/\epsilon^2)$ for MLA.

Let us compare MLA with ULA (i.e., MLA in the Euclidean case with $\phi(x) = \frac{1}{2}\|x\|^2$).
Recall for ULA, the mean-square analysis by~\cite{LZT21} yields a biased convergence guarantee where the bias scales as $O(\sqrt{d} h)$ under an additional 3rd-order regularity condition on $f$.
This leads to a mixing time bound of $\tilde O(\sqrt{d}/\epsilon)$ for ULA.
We see the bias of MLA has a worse dependence on $h$ than the bias of ULA.
This is because the continuous-time Mirror Langevin Dynamics~\eqref{Eq:MLD} of MLA has a changing covariance, while the usual continuous-time Langevin Dynamics of ULA has a constant covariance; therefore, MLA incurs an additional stochastic error from the Brownian motion part, which is not incurred by ULA.
Formally, this is reflected in the orders of error of the two algorithms: We show below that MLA has local weak and strong errors of orders $p_1 =\frac{3}{2}$ at least and $p_2 = 1$ (note the local weak order of MLA is actually $p_1=2$, because it is the Euler-Maruyama discretization of an SDE; the multiplicative noise causes the strong error to lose half an order, but not the weak error (see e.g., \cite[page 14]{milstein2013stochastic}); however, we will see that as long as $p_1\geq p_2+\frac{1}{2}$, the order of the final sampling error is determined by $p_2$ but not $p_1$, and even though our $p_1=\frac{3}{2}$ bound is not tight in order, its constants can be made very explicit and hence helpful to later analysis). On the other hand, it is well known that ULA has local weak and strong error of orders $p_1 = 2$ and $p_2 = \frac{3}{2}$ because it is the Euler-Maruyama discretization of an SDE with additive noise (see \cite{milstein2013stochastic} for the general theory and \cite{LZT21} for details of worked out constants).
It would be interesting to understand whether we can improve the local errors and the bias of MLA, perhaps using more sophisticated discretization of MLD to improve the stochastic error.

Our result improves on the analysis of~\cite{ZPFP20}, who assume stronger assumptions (our assumptions {\bf (A1)}, {\bf (A2)}, {\bf (A3)}, along with two assumptions on the moment of $\nabla^2 \phi$ and a bound on the commutator of $\nabla^2 f$ and $\nabla^2 \phi$), and prove a biased convergence analysis where the bias scales as $O(\sqrt{d h} + r_0)$, where $r_0 = O(\sqrt{\alpha d})$ does not depend on $h$.
Note in the Euclidean case (when $\phi(x) = \frac{1}{2}\|x\|^2$), the modified self-concordance parameter is $\alpha = 0$, and thus $r_0 = 0$; 
but for general $\phi$, the asymptotic radius is positive: $r_0 > 0$, so the result of~\cite{ZPFP20} does not guarantee convergence to the target.
With our mean-square analysis, we have shown that in fact there is no dependence on this radius $r_0$, and the bias indeed scales as $O(\sqrt{d h})$.

We note our result uses the modified self-concordance property, as also used in~\cite{ZPFP20}. 
In one-dimension ($d=1$), modified self-concordance is equivalent to the classical self-concordance property: Both are equivalent to the condition that $x \mapsto 1/\sqrt{\phi''(x)}$ is a Lipschitz function.
However, in higher dimension, they are different.
In particular, modified self-concordance is not an affine-invariant property (in contrast to the classical self-concordance), and the parameter $\alpha$ can be arbitrarily large; see example in Appendix~\ref{App:LogBarrier}.
This is problematic since our convergence bound only holds when $\alpha$ is less than $m$ (the strong convexity parameter).
It would be desirable to have an analysis of MLA with the more natural self-concordance property.

Our result in Theorem~\ref{Thm:Main} shows that to obtain a consistent algorithm (with a vanishing bias) from MLD, it suffices to apply a simple discretization such as MLA.
This shows we do not need to use an exact simulator of the Brownian motion with changing covariance, as proposed by~\cite{AC20}, which allows a nice analysis under self-concordance property.
It would be interesting to bridge the analysis technique to MLA.

The relative smoothness {\bf (A2)} and relative strong convexity {\bf (A3)} conditions imply that the Hessian of $f$ are bounded by the Hessian of $\phi$:
\begin{align*}
m \nabla^2 \phi(x) \preceq \nabla^2 f(x) \preceq M \nabla^2 \phi(x) ~~~~~~ \forall ~ x \in \X.
\end{align*}
See~\cite[Appendix~B]{ZPFP20} for more details.
Since we assume $\phi$ is a Legendre function, $\nabla^2 \phi(x) \to \infty$ as $x \to \partial \X$;
then for our result to hold, we need $\nabla^2 f \to \infty$ as $x \to \partial \X$.
This restricts the applicability of the result; for example, it does not apply when $\nu$ is a uniform ($f = 0$) or Gaussian distribution ($f$ is quadratic) restricted on a polytope with $\phi$ being the log-barrier function.
It is desirable to have a more general convergence analysis of MLA under weaker conditions on $f$ and $\phi$.

\section{Proof of main result}

The proof of Theorem~\ref{Thm:Main} uses the mean-square analysis framework described in~\cite{LZT21}. 
We review the mean-square analysis framework in Section~\ref{Sec:MeanSquare}.
We verify the conditions hold for MLA in Section~\ref{Sec:MLACheck}, and apply the mean-square analysis to prove Theorem~\ref{Thm:Main} in Section~\ref{Sec:ProofThmMain}.

\subsection{A review of the mean-square analysis framework}
\label{Sec:MeanSquare}

Mean-square analysis was a classical tool for analyzing the integration error of SDEs (e.g., \cite{milstein2013stochastic}).~\cite{li2019stochastic} extended it to obtain non-asymptotic sampling error bound of an algorithm which is a discretization of a decaying stochastic differential equation (SDE). While~\cite{li2019stochastic} required the local errors to satisfy uniform bounds,~\cite{LZT21} relaxes this requirement and only needs non-uniform bounds. We will establish non-uniform local error bounds for MLA, and thus use the version of mean-square analysis in~\cite{LZT21}. The results will be reviewed in a simplified setting; see~\cite[Section~3]{LZT21} for details.

\paragraph{Contractive SDE.}
Consider a continuous-time process $Y_t \in \R^d$ which evolves following the SDE:
\begin{align}\label{Eq:SDE}
dY_t = -g(Y_t) \, dt + \sqrt{2} A(Y_t) \, dW_t
\end{align}
for some vector field $g \colon \R^d \to \R^d$ and matrix $A \colon \R^d \to \R^{d \times d}$.
We assume $g$ and $A$ are Lipschitz continuous.
Here $W_t$ is the standard Brownian motion in $\R^d$.

We say the SDE~\eqref{Eq:SDE} is {\bf contractive} with rate $\beta > 0$ if there exists $t_0>0$ such that any two solutions $Y_t, Y_t'$ with synchronous coupling (i.e.\  driven by the same Brownian motion) satisfy:
\begin{align}
\E[\|Y_t-Y_t'\|^2] \le e^{-2\beta t} \E[\|Y_0 - Y_0'\|^2] ~~~~~~ \forall ~ t \in (0,t_0).
\end{align}
If the SDE~\eqref{Eq:SDE} is contractive, then it has a stationary distribution $\tilde \nu$.

\paragraph{Short-time deviation.}
Since $g$ and $A$ are Lipschitz continuous,
one can show~\cite[Lemma~1.3]{milstein2013stochastic} that there exist a maximum time $t_{0} > 0$ and a constant $C_0 > 0$ such that for any solutions $Y_t, Y_t'$ with synchronous coupling:
\begin{align}
\E[\|(Y_t' - Y_0') - (Y_t - Y_0)\|^2_2] \le C_0 \, \E[\|Y_0'-Y_0\|^2_2] \, t ~~~~~~ \forall ~ 0 < t \le t_{0}.
\end{align}

\paragraph{Algorithm and local error.}
Suppose we have an algorithm $\Alg_h$ depending on a step size $h > 0$
that simulates the solution $Y_t$ of the SDE~\eqref{Eq:SDE} at time $t = h$.

For any $Y_0 \in \R^d$, let $Y_h$ denote the solution of the SDE~\eqref{Eq:SDE} at time $t = h$,
and let $\bar Y_1 = \Alg_h(Y_0)$ denote the output of the algorithm from $Y_0$.
We say that the algorithm has (non-uniform) {\bf local weak error} of order $p_1$ if there exist a maximum step size $h_1 > 0$ and constants $C_1, D_1 \ge 0$ such that
\begin{align}
\| \E[Y_h - \bar Y_1] \| \le \left(C_1 + D_1 \sqrt{\E[\|Y_0\|^2} \right) h^{p_1} ~~~~~~ \forall ~ 0 < h \le h_1.
\end{align}
We say the algorithm has (non-uniform) {\bf local strong error} of order $p_2$ if there exist a maximum step size $h_2 > 0$ and constants $C_2, D_2 \ge 0$ such that
\begin{align}
\E[\| Y_h - \bar Y_1\|^2] \le \left(C_2^2 + D_2^2 \E[\|Y_0\|^2 \right) h^{2p_2} ~~~~~~ \forall ~ 0 < h \le h_2.
\end{align}
Here $Y_h$ and $\bar Y_1$ are coupled by sharing the same filtration (i.e.\ the algorithm $\Alg_h$ has access to the realization of the Wiener process that generates $Y_h$).

When $D_1=D_2=0$, the bounds are termed as {\em uniform bounds} in~\cite{li2019stochastic}.

\paragraph{Bound on global error.}
With the set-up above, the mean-square analysis framework produces the following bound on the global (long-term) error.

\begin{theorem}[{{\cite[Theorem~3.3,~3.4]{LZT21}}}]\label{Thm:MeanSq}
Assume the SDE~\eqref{Eq:SDE} is contractive with rate $\beta > 0$.
Assume the algorithm $\Alg_h$ has local weak error of order $p_1$ and local strong error of order $p_2$ with $\frac{1}{2} < p_2 \le p_1 -\frac{1}{2}$.
Let us define a maximum step size $h_{\max} > 0$ by
\begin{align}
h_{\max} = \min\left\{ t_0, h_1, h_2, \frac{1}{4\beta}, \left(\frac{\sqrt{\beta}}{4\sqrt{2} D_2}\right)^{\frac{1}{p_2-\frac{1}{2}}}, \, \left(\frac{\beta}{8\sqrt{2} (D_1 + C_0 D_2)}\right)^{\frac{1}{p_2-\frac{1}{2}}} \right\}
\end{align}
and constants $U = \sqrt{4 \E[\|Y_0\|^2] + 6 \E_{\tilde \nu}[\|Y\|^2]}$ and $C > 0$ by
\begin{align}
C = \frac{2}{\sqrt{\beta}} \left(\frac{C_1 + C_0 C_2 + \sqrt{2} U (D_1 + C_0 D_2)}{\sqrt{\beta}} + C_2 + \sqrt{2} D_2 U \right).
\end{align}
Starting from any $Y_0 = \bar Y_0 \sim \tilde \rho_0$, suppose we run the algorithm $\Alg_h$ with step size $0 < h \le h_{\max}$ to produce iterates $\bar Y_k = \Alg_h(\bar Y_{k-1}) \sim \tilde \rho_k$.
Let $Y_{hk}$ denote the solution to the SDE~\eqref{Eq:SDE} at time $t = hk$.
Then $\bar Y_k$ is close to $Y_{hk}$ at all time:
\begin{align}
\sqrt{\E[\|Y_{hk} - \bar Y_k\|^2]} \le C h^{p_2-\frac{1}{2}} ~~~~~~ \forall ~ k \ge 0.
\end{align}
Furthermore, the distribution of $\bar Y_k \sim \tilde \rho_k$ has the following biased convergence guarantee:
\begin{align}
W_2(\tilde \rho_k, \tilde \nu) \le \sqrt{2} e^{-\beta k h} W_2(\tilde \rho_0, \tilde \nu) + \sqrt{2} C h^{p_2 - \frac{1}{2}} ~~~~~~ \forall ~ k \ge 0.
\end{align}
\end{theorem}

\subsection{Application to MLA}
\label{Sec:MLACheck}

For our sampling problem, 
we wish to apply the mean-square analysis framework to the Mirror Langevin Algorithm in the dual space~\eqref{Eq:MLA2}.
The continuous-time SDE~\eqref{Eq:SDE} of MLA is the Mirror Langevin Dynamics, which we review in the next section.
We establish the local error orders of MLA in the following section.

\subsubsection{Mirror Langevin Dynamics}
\label{Sec:MLD}

Consider the {\bf Mirror Langevin Dynamics (MLD)}, which is a stochastic process $Y_t \in \R^d$ following the SDE:
\begin{align}\label{Eq:MLD}
    dY_t = -g(Y_t) dt + \sqrt{2} A(Y_t) dW_t
\end{align}
where as defined in~\eqref{Eq:g} and \eqref{Eq:A}, 
$g(y) = \nabla f (\nabla \phi^\ast(y))$ and
$A(y) = \sqrt{\nabla^2 \phi^\ast(y)^{-1}}$.
The stationary distribution of MLD~\eqref{Eq:MLD} is the target distribution in the dual space: $\tilde \nu = (\nabla \phi)_\# \nu$.

By assumptions~{\bf (A1)} and {\bf (A2)}, $g$ and $A$ are Lipschitz continuous.
Let us establish the contractivity and deviation bound on MLD.
The proofs are provided in Appendix~\ref{Sec:Proofs}.

\begin{lemma}\label{Lem:Contraction}
Assume~{\bf (A1)} and {\bf (A2)} with $\alpha < m$.
Then MLD~\eqref{Eq:MLD} is contractive with rate $\beta = m - \alpha$.
\end{lemma}

\begin{lemma}\label{Lem:Deviation}
Assume~{\bf (A1)}, {\bf (A2)}, and {\bf (A3)} with $\alpha < m$.
Then any two solutions $Y_t, Y_t'$ of MLD~\eqref{Eq:MLD} with synchronous coupling satisfy
\begin{align}
    \E[\|(Y_t' - Y_0') - (Y_t - Y_0)\|^2] \le 4M \, \E[\|Y_0'-Y_0\|^2] \, t ~~~~~~ \forall ~ t \ge 0.
\end{align}
\end{lemma}

We also need the following bound on MLD.
Let $x^\ast = \arg\min_{x \in \X} f(x)$ and $y^\ast = \nabla \phi(x^*) \in \R^d$.

\begin{lemma}\label{Lem:Growth}
Assume~{\bf (A1)}, {\bf (A2)}, and {\bf (A3)}. 
Along MLD~\eqref{Eq:MLD}, for $0 < t \le \frac{1}{M^2+4\alpha}$,
\begin{align}\label{Eq:Growth}
    \E[\|Y_t-Y_0\|^2] \le \gamma \, t
\end{align}
where $\gamma = 8(1 + 4\alpha) \E[\|Y_0\|^2] + 8(1+4\alpha) \|y^\ast\|^2 + 16\|A(y^\ast)\|^2_{\HS} + \frac{4}{M^2} \|g(y^\ast)\|^2$.
\end{lemma}

\begin{remark}
In Lemma~\ref{Lem:Contraction} we show MLD is contracting if $\alpha < m$.
In general, a bound on $\alpha$ (the Lipschitz constant of the covariance) is necessary for an SDE with multiplicative noise to contract; see the example of the geometric Brownian motion in Appendix~\ref{App:GBM}.
\end{remark}

\subsubsection{Local Errors of the Mirror Langevin Algorithm}

Let us now consider the algorithm $\Alg_h$ to be the Mirror Langevin Algorithm in the dual space~\eqref{Eq:MLA2}.
We can show MLA has the following local errors.
The proofs are provided in Appendix~\ref{Sec:Proofs}.

\begin{lemma}\label{Lem:WeakError}
Assume~{\bf (A1)}, {\bf (A2)}, and {\bf (A3)}.
Then MLA~\eqref{Eq:MLA2} has local weak error at least of order $p_1 = \frac{3}{2}$, with maximum step size $h_1 = \frac{1}{M^2 + 4\alpha}$ and constants
\begin{align*}
C_1 &= 3M \sqrt{1 + 4\alpha} \,  \left( \|y^\ast\| + \|A(y^\ast)\|_{\HS} + \frac{1}{M} \|g(y^\ast)\| \right) \\
D_1 &= 2M \sqrt{1 + 4\alpha}.
\end{align*}
\end{lemma}

\begin{lemma}\label{Lem:StrongError}
Assume~{\bf (A1)}, {\bf (A2)}, and {\bf (A3)}.
Then MLA~\eqref{Eq:MLA2} has local strong error at least of order $p_2 = 1$, with maximum step size $h_2 = \frac{1}{M^2 + 4\alpha}$ and constants
\begin{align*}
C_2 &= 7(1+4\alpha) \left( \|y^\ast\| + \|A(y^\ast)\|_{\HS} + \frac{1}{M} \|g(y^\ast)\| \right) \\
D_2 &= 5 (1+4\alpha).
\end{align*}
\end{lemma}

\subsection{Proof of Theorem~\ref{Thm:Main}: Convergence Rate of MLA}
\label{Sec:ProofThmMain}

\begin{proof}[Proof of Theorem~\ref{Thm:Main}]
Assume {\bf (A1)}, {\bf (A2)}, and {\bf (A3)} with $\alpha < m$.
We have verified that MLA satisfies the conditions in the mean-square analysis framework:
In Lemma~\ref{Lem:Contraction} we show MLD is contractive with rate $\beta = m - \alpha$.
We derive the deviation bound in Lemma~\ref{Lem:Deviation} with $C_0 = 4M$.
In Lemmas~\ref{Lem:WeakError} and~\ref{Lem:StrongError} we show MLA has local weak error of order $p_1 = \frac{3}{2}$ and local strong error of order $p_2 = 1$, and indeed $p_2 \le p_1 - \frac{1}{2}$.

Then by Theorem~\ref{Thm:MeanSq}, we can compute the maximum step size:
\begin{align*}
h_{\max} &= \min\left\{ \frac{1}{M^2+4\alpha}, \frac{1}{4\beta},  \left(\frac{\sqrt{\beta}}{4\sqrt{2} D_2}\right)^{\frac{1}{p_2-\frac{1}{2}}}, \, \left(\frac{\beta}{8\sqrt{2} (D_1 + C_0 D_2)}\right)^{\frac{1}{p_2-\frac{1}{2}}} \right\} \\
&= \min\Bigg\{\frac{1}{M^2+4\alpha}, \frac{1}{4(m-\alpha)}, \frac{m-\alpha}{800(1 + 4\alpha)^2},
 \frac{(m-\alpha)^2}{128 \left( 2M \sqrt{(1 + 4\alpha)} + 20M (1 + 8\alpha) \right)^2} \Bigg\} \\
&= \O\left(\frac{(m-\alpha)^2}{M^2(1+4\alpha)^2}\right).
\end{align*}

Recall $\tilde \nu = (\nabla \phi)_\# \nu$ is the target distribution of MLD~\eqref{Eq:MLD}.
We can compute the constant
\begin{align*}
    U = \sqrt{4 \E[\|Y_0\|^2_2] + 6 \E_{\tilde \nu}[\|Y\|^2_2]} = O(\sqrt{d}).
\end{align*}
Note that $\|A(y^*)\|_{\HS} = \sqrt{\Tr(A(y^*) A(y^*)^\top)} = \sqrt{\Tr(\nabla^2 \phi^*(y^*)^{-1})} = \sqrt{\Tr(\nabla^2 \phi(x^*))} = O(\sqrt{d})$.
Let us define $V := \|y^\ast\| + \|A(y^\ast)\|_{\HS} + \frac{1}{M} \|g(y^\ast)\| = O(\sqrt{d})$.
Then the resulting constant is
\begin{align*}
C_{\MLA} &= \frac{2}{\beta} \left(C_1 + C_0 C_2 + \sqrt{2} U(D_1 + C_0 D_2) \right) + \frac{2}{\sqrt{\beta}} \left( C_2 + \sqrt{2} D_2 U \right) \\
&= \frac{2}{m-\alpha} 
\Bigg(3M \sqrt{(1 + 4\alpha)} V + 28M (1+4\alpha) V + \sqrt{2} U \left( 2 M \sqrt{(1 + 4\alpha)} + 20M (1 + 4\alpha) \right) \Bigg) \\
&~~~~ + \frac{2}{\sqrt{m-\alpha}} \left( 7(1+4\alpha) V + 5\sqrt{2} (1 + 4\alpha) U \right) \\
&= \O\left( \frac{M (1+4\alpha) \sqrt{d} }{m-\alpha} \right).
\end{align*}
The conclusion of Theorem~\ref{Thm:Main} follows from Theorem~\ref{Thm:MeanSq}.
\end{proof}

\section{Discussion}

Our result leaves open many questions, including the following.
It would be interesting to consider a more sophisticated discretization of MLD 
such that the mean-square analysis framework will show improved local errors and smaller bias.

It would be interesting to have a better analysis of MLA under more natural conditions on $\phi$, such as self-concordance (rather than modified self-concordance), and under relaxed requirements on $f$ and $\phi$ (e.g.\ that allows us to sample from a uniform or Gaussian distribution on a polytope).
It would be desirable to have a convergence analysis of MLA in the Wasserstein distance generated by the Hessian metric $\nabla^2 \phi$ rather than the Euclidean metric, or in other measures such as KL or $\chi^2$-divergence.

It would be interesting to understand whether we can discretize the Newton Langevin Dynamics (which is the case when $\phi = f$ as described in Appendix~\ref{App:NLD} and which is affine-invariant in continuous time) and obtain a discrete-time algorithm with a convergence guarantee which is also affine-invariant.

It would also be interesting to understand whether we can derive a more general discrete-time analysis framework that works under a relaxed condition, e.g.\ without requiring contraction in continuous time, but only exponential convergence in function value (which is known for ULA under the log-Sobolev inequality, see for example~\cite{VW19}).

\newpage
\bibliography{reference}

\newpage
\appendix

\section{Riemannian and Mirror Langevin Dynamics in Continuous Time}
\label{App:CtsTime}

Consider the problem of sampling from $\nu \propto e^{-f}$ on $\X \subseteq \R^d$ as described in Section~\ref{Sec:SetUp}.

Suppose we endow $\X$ with a Riemannian metric $\g$, which we write as a positive definite matrix: $\g(x) \succ 0$ for all $x \in \X$.
This means at each point $x \in \X$ we measure local norm using the metric $\g(x)$:
$$\langle u,v \rangle_x := u^\top \g(x) v$$
for all $u,v$ in the tangent space. 
We assume $x \mapsto \g(x)$ is differentiable.
Let $M(x) = \g(x)^{-1}$ be the inverse matrix, and let $\sqrt{M(x)}$ be a square-root of $M(x)$.
Let $\nabla \cdot M(x) \in \R^d$ be the {\em divergence} of $M$, which is a vector-valued function whose entries are the divergences of the columns of $M$.
We assume $\g(x) \to \infty$ (equivalently, $M(x) \to 0$) as $x$ approaches the boundary of $\X$.

\subsection{Review for optimization}

Recall in optimization, the {\bf Riemannian gradient flow (RGF)} (or natural gradient flow) for minimizing $f$ using the metric $\g$ is the solution $X_t$ to the differential equation:
$$\dot X_t = \frac{d}{dt} X_t = -M(X_t) \, \nabla f(X_t).$$
Here we use the inverse metric $M(x) = \g(x)^{-1}$ to turn the $\ell_2$-gradient $\nabla f(x) = (\part{f(x)}{x_1},\dots,\part{f(x)}{x_d})$ into a gradient tangent vector $\grad \, f(x) = M(x) \nabla f(x)$ under the Riemannian metric $\g(x)$.
RGF has nice properties when the objective function $f$ satisfies some properties.
For example, if $f$ is {\em geodesically strongly convex} (which means $f$ is strongly convex along geodesics generated by the Riemannian metric $\g$), then RGF is exponentially contracting.
Moreover, if $f$ is {\em gradient dominated} with respect to $\g$,
then the function value $f(X_t)$ converges exponentially fast along RGF. 

Consider when the metric $\g(x)$ is given by the Hessian of a convex Legendre function $\phi$: $\g(x) = \nabla^2 \phi(x) \succ 0$.
Then the RGF becomes:
$$\dot X_t = -\nabla^2 \phi(X_t)^{-1} \, \nabla f(X_t).$$
In terms of the dual variable $Y_t = \nabla \phi(X_t)$, this becomes the {\bf mirror flow}:
$$\dot Y_t = -\nabla f(X_t) = -\nabla f(\nabla \phi^*(Y_t)).$$
Recall by the mirror map $\nabla \phi$, the metric $\nabla^2 \phi$ on $\X$ becomes the Hessian metric $\nabla^2 \phi^*$ on $\Y = \nabla \phi(X) = \R^d$.
The mirror flow is also the Riemannian gradient flow for minimizing the pushforward function $\tilde f(y) = f(\nabla \phi^*(y))$ under the Hessian metric $\nabla^2 \phi^*(y)$ (because $\grad \, \tilde f(y) = \nabla^2 \phi^*(y)^{-1} \nabla \tilde f(y) = \nabla f(\nabla \phi^*(y))$).
Discretizing the mirror flow gives the {\em mirror descent} algorithm in optimization.

\subsection{Riemannian Langevin Dynamics}

The {\bf Riemannian Langevin Dynamics (RLD)} for sampling from $\nu \propto e^{-f}$ on $\X$ using the metric $\g(x)$ is the solution $X_t$ to the stochastic differential equation:
\begin{align}\label{Eq:RLD}
dX_t = \left(\nabla \cdot M(X_t) - M(X_t) \, \nabla f(X_t)\right) dt + \sqrt{2} \sqrt{M(X_t)} \, dW_t.
\end{align}
Here $W_t$ is the standard Brownian motion in $\R^d$.
Since $M(x) \to 0$ as $x \to \partial \X$, the process does not leave $\X$: If $X_0 \in \X$, then $X_t \in \X$ for all $t > 0$.

The additional drift term $\nabla \cdot M(X_t)$ accounts for the covariance $M(X_t)$ in the Brownian motion.
The stationary distribution for RLD is $\nu(x) \propto e^{-f(x)}$ (the density is with respect to the Lebesgue measure $dx$ on $\R^d$).
This can be seen, for example, from the following Fokker-Planck equation.

If $X_t \in \X$ follows RLD~\eqref{Eq:RLD}, then its density $\rho_t \colon \X \to \R$ evolves following the partial differential equation (PDE):
\begin{align}\label{Eq:WLD_FP}
\part{\rho_t}{t} = \nabla \cdot \left( \rho_t M \nabla \log \frac{\rho_t}{\nu} \right). 
\end{align}
Clearly if $\rho_t = \nu$ then the dynamics is stationary.
Furthermore, the PDE above can be interpreted as the gradient flow for minimizing relative entropy with respect to the Wasserstein metric on the metric space $(\X,\g)$.

From the Fokker-Planck equation~\eqref{Eq:WLD_FP}, we can derive how fast the dynamics RLD approaches the target distribution $\nu$ in various measures.

For example, recall the {\em $\chi^2$-divergence} of a probability distribution $\rho$ with respect to $\nu$ is
\begin{align*}
\chi^2_\nu(\rho) = \Var_\nu\left(\frac{\rho}{\nu}\right) = \int_\X \nu(x) \left(\frac{\rho(x)}{\nu(x)}-1\right)^2 dx
= \int_\X \frac{\rho(x)^2}{\nu(x)} dx - 1.
\end{align*}
Then a standard calculation reveals that
the $\chi^2$-divergence is decreasing along RLD~\eqref{Eq:WLD_FP}:
\begin{align*}
\frac{d}{dt} \chi^2_\nu(\rho_t) = -2G_\nu(\rho_t)
\end{align*}
where
\begin{align*}
G_\nu(\rho) = \E_\nu\left[\left\| \nabla\left(\frac{\rho}{\nu}\right)\right\|^2_{M}\right]
= \int_{\X} \nu(x) \left\langle \nabla \left(\frac{\rho(x)}{\nu(x)}\right), M(x) \nabla \left(\frac{\rho(x)}{\nu(x)}\right) \right\rangle \, dx.
\end{align*}
Therefore, if $\nu$ satisfies a {\em Poincar\'e inequality} with respect to $\g$, which means for any differentiable function $h \colon \X \to \R$ we have
\begin{align}\label{Eq:Poincare}
    \Var_\nu(h) \le C_{\mathrm{P}} \, \E_\nu[\|\nabla h\|^2_M],
\end{align}
then we can conclude RLD converges exponentially fast in $\chi^2$-divergence: $\chi^2_\nu(\rho_t) \le e^{-\frac{2}{C_{\mathrm{P}}} t} \chi^2_\nu(\rho_0)$.

Similarly, recall the {\em relative entropy} (or {\em KL divergence}) of $\rho$ with respect to $\nu$ is
\begin{align*}
    H_\nu(\rho) = \E_\nu\left[ \frac{\rho}{\nu} \log \frac{\rho}{\nu} \right]
    = \int_\X \rho(x) \log \frac{\rho(x)}{\nu(x)} \, dx. 
\end{align*}
Then along RLD~\eqref{Eq:WLD_FP}, KL divergence is decreasing:
\begin{align*}
    \frac{d}{dt} H_\nu(\rho_t) = -J_\nu(\rho_t)
\end{align*}
where $J_\nu(\rho)$ is the relative Fisher information of $\rho$ with respect to $\nu$ under the metric $\g$:
\begin{align*}
    G_\nu(\rho) = \E_\rho\left[\left\|\nabla \log \frac{\rho}{\nu} \right\|^2_M\right].
\end{align*}
Therefore, if $\nu$ satisfies a {\em log-Sobolev inequality} with respect to $\g$, which means for any $\rho$ we have $$H_\nu(\rho) \le C_{\mathrm{LSI}} \, J_\nu(\rho),$$
then we can conclude RLD converges exponentially fast in KL divergence: $H_\nu(\rho_t) \le e^{-\frac{t}{C_{\mathrm{LSI}}}} H_\nu(\rho_0)$.

\subsection{Mirror Langevin Dynamics}

Suppose now the metric $\g(x)$ is given by the Hessian of a convex Legendre function $\phi$: $\g(x) = \nabla^2 \phi(x) \succ 0$.
The Riemannian Langevin dynamics~\eqref{Eq:RLD} becomes the following SDE, which is also studied by~\cite{ZPFP20} and~\cite{ChewiEtAl20}:
\begin{align}\label{Eq:WLD_M}
dX_t = \left(\nabla \cdot (\nabla^2 \phi(X_t)^{-1}) - \nabla^2 \phi(X_t)^{-1} \, \nabla f(X_t)\right) dt + \sqrt{2 \nabla^2 \phi(X_t)^{-1}} \, dW_t.
\end{align}
If $\nu$ satisfies log-Sobolev or Poincar\'e inequality (which is called {\em mirror Poincar\'e inequality} in~\cite{ChewiEtAl20}), then we can conclude exponential convergence rate in KL or $\chi^2$ divergence along~\eqref{Eq:WLD_M}.

The SDE~\eqref{Eq:WLD_M} requires $\nabla \cdot (\nabla^2 \phi(x)^{-1})$, which may be complicated.
Consider the dual variable $Y_t = \nabla \phi(X_t)$.
By It\^o's lemma, $Y_t$  
evolves following the {\bf Mirror Langevin Dynamics}:
\begin{align*}
dY_t = -\nabla f(\nabla \phi^\ast(Y_t)) \, dt + \sqrt{2 \nabla^2 \phi^\ast(Y_t)^{-1}} \, dW_t.
\end{align*}
In particular, the drift term simplifies and there is no third derivative involved.
The mirror Langevin dynamics is also the Riemannian Langevin dynamics~\eqref{Eq:RLD} for sampling from the pushforward distribution $\tilde \nu = (\nabla \phi)_\# \nu$ using the Hessian metric $\nabla^2 \phi^*$.
Furthermore, the $\chi^2$-divergence and KL divergence are invariant under the mirror map.
Therefore, $\nu$ satisfies LSI or Poincar\'e inequality with respect to $\nabla^2 \phi$ if and only if $\tilde \nu$ also satisfies LSI or Poincar\'e inequality with respect to $\nabla^2 \phi^*$.
Therefore, we get the same convergence guarantee in both primal and dual spaces.

\subsection{Newton Langevin Dynamics}
\label{App:NLD}

A particularly nice choice of $\phi$ is when $\phi = f$.
This gives the {\bf Newton Langevin Dynamics}, which in the primal space takes the form:
\begin{align}\label{Eq:NLD}
dX_t = \left(\nabla \cdot (\nabla^2 f(X_t)^{-1}) - \nabla^2 f(X_t)^{-1} \, \nabla f(X_t)\right) dt + \sqrt{2} \sqrt{\nabla^2 f(X_t)^{-1}} \, dW_t.
\end{align}
A remarkable property of NLD, as pointed out by~\cite{ChewiEtAl20}, is that the Poincar\'e inequality of $\nu \propto e^{-f}$ with respect to its Hessian metric $\nabla^2 f$ is always true with a uniform constant $C_{\mathrm{P}} = 1$ for {\em any} strictly log-concave distribution $\nu$, by the virtue of the Brascamp-Lieb inequality~\citep{BL76}.
This gives a uniform exponential convergence rate along NLD in $\chi^2$-divergence as well as the Wasserstein distance with respect to the metric $\nabla^2 f$;
see detailed exposition and additional consequences in~\cite{ChewiEtAl20}.

In the dual space, Newton Langevin Dynamics has a simple drift:
\begin{align}\label{Eq:NLD2}
dY_t = -Y_t \, dt + \sqrt{2 \nabla^2 f^\ast(Y_t)^{-1}} \, dW_t
\end{align}
since $\nabla f(\nabla f^*(y)) = y$.
The target distribution of NLD in the dual space is the pushforward distribution $\tilde \nu = (\nabla f)_{\#} \nu$ where $\nu \propto e^{-f}$.
The SDE~\eqref{Eq:NLD2} for sampling from $\tilde \nu$ was also pointed out by~\cite{Fathi19} from the study of Stein's kernel.

\section{Proofs of Lemmas}
\label{Sec:Proofs}

\subsection{Proof of Lemma~\ref{Lem:Contraction}: Contraction of MLD}

\begin{proof}[Proof of Lemma~\ref{Lem:Contraction}]
Assume {\bf (A1)} and {\bf (A2)}.
We will show MLD~\eqref{Eq:MLD} is contractive if $\alpha < \frac{m}{2}$.

Suppose we have two solutions $Y_t', Y_t$ of MLD~\eqref{Eq:MLD} with the same Brownian motion:
\begin{align*}
    dY_t' &= -g(Y_t') dt + \sqrt{2} A(Y_t') dW_t \\
    dY_t &= -g(Y_t) dt + \sqrt{2} A(Y_t) dW_t.
\end{align*}
Then the difference satisfies the SDE
\begin{align}\label{Eq:MLD_Coupling}
    d(Y_t' - Y_t) &= -(g(Y_t') - g(Y_t)) dt + \sqrt{2} (A(Y_t') - A(Y_t)) dW_t.
\end{align}

Recall in general that if $V_t \in \R^d$ follows a general SDE $dV_t = b(V_t) dt + G(V_t) dW_t$, then
\begin{align*}
    \frac{d}{dt} \E[\|V_t\|^2] = \E[2\langle b(V_t), V_t \rangle + \|G(V_t)\|^2_{\HS}].
\end{align*}
Then for the SDE~\eqref{Eq:MLD_Coupling} of the difference $V_t = Y_t' - Y_t$, and by applying assumptions {\bf (A1)} and {\bf (A3)}, we have
\begin{align*}
    \frac{d}{dt} \E[\|Y_t'-Y_t\|^2]
    &= -2 \E[\langle g(Y_t') - g(Y_t), Y_t' - Y_t \rangle] + 2\E[\|A(Y_t') - A(Y_t)\|^2_{\HS}] \\
    &\le -2m \E[\|Y_t'-Y_t\|^2_2] + 2\alpha \E[\|Y_t'-Y_t\|^2_2] \\
    &= -2(m-\alpha) \E[\|Y_t'-Y_t\|^2_2].
\end{align*}
We see that we have an exponential contraction if $\alpha < m$:
\begin{align}\label{Eq:Contraction}
    \E[\|Y_t'-Y_t\|^2] \le \exp\left(-2(m-\alpha) t \right) \E[\|Y_0'-Y_0\|^2] ~~~~~~ \forall ~ t \ge 0.
\end{align}
This shows that MLD~\eqref{Eq:MLD} is contractive with rate $\beta = m-\alpha$.
\end{proof}

\subsection{Proof of Lemma~\ref{Lem:Deviation}: Deviation bound of MLD}

\begin{proof}[Proof of Lemma~\ref{Lem:Deviation}]
Assume {\bf (A1)},  {\bf (A2)}, and {\bf (A3)} with $\alpha < m$.

Suppose we have two solutions $Y_t', Y_t$ of MLD~\eqref{Eq:MLD} with the same Brownian motion:
\begin{align*}
    dY_t' &= -g(Y_t') dt + \sqrt{2} A(Y_t') dW_t \\
    dY_t &= -g(Y_t) dt + \sqrt{2} A(Y_t) dW_t.
\end{align*}
Consider the shifted variables $\tilde Y_t' = Y_t' - Y_0'$ and $\tilde Y_t = Y_t - Y_0$, which satisfy:
\begin{align*}
    d\tilde Y_t' &= -g(\tilde Y_t' + Y_0') dt + \sqrt{2} A(\tilde Y_t' + Y_0') dW_t \\
    d\tilde Y_t &= -g(\tilde Y_t + Y_0) dt + \sqrt{2} A(\tilde Y_t + Y_0) dW_t.
\end{align*}
Then the difference $\tilde Y_t' - \tilde Y_t = (Y_t' - Y_0') - (Y_t - Y_0)$ satisfies:
\begin{align*}
    d(\tilde Y_t' - \tilde Y_t)
    &= -(g(\tilde Y_t' + Y_0') - g(\tilde Y_t + Y_0)) dt + \sqrt{2} (A(\tilde Y_t' + Y_0')-A(\tilde Y_t + Y_0)) dW_t.
\end{align*}

By Lemma~\ref{Lem:Contraction}, we have the contraction result~\eqref{Eq:Contraction}, which implies $\E[\|Y_t'-Y_t\|^2_2] \le \E[\|Y_0'-Y_0\|^2_2]$ for all $t \ge 0$.
Then by applying {\bf (A1)} and {\bf (A2)} and using $\alpha < m \le M$, we get
\begin{align*}
    \frac{d}{dt} &\E[\|(Y_t' - Y_0') - (Y_t - Y_0)\|^2_2] \\
    &= \frac{d}{dt} \E[\|\tilde Y_t' - \tilde Y_t\|^2_2] \\
    &= -2\E[\langle g(\tilde Y_t' + Y_0') - g(\tilde Y_t + Y_0), \tilde Y_t' - \tilde Y_t \rangle]
    + 2 \E[\|A(\tilde Y_t' + Y_0')-A(\tilde Y_t + Y_0)\|^2_{\HS}] \\
    &= -2\E[\langle g(Y_t') - g(Y_t), Y_t' - Y_t - (Y_0' - Y_0) \rangle]
    + 2 \E[\|A(Y_t')-A(Y_t)\|^2_{\HS}] \\
    &\le 2\E[\langle g(Y_t') - g(Y_t), Y_0' - Y_0 \rangle]
    + 2\alpha \E[\|Y_t'-Y_t\|^2_2]  \\
    &\le 2\E[\| g(Y_t') - g(Y_t) \|^2_2]^{1/2} \E[\|Y_0' - Y_0\|^2_2]^{\frac{1}{2}}
    + 2\alpha \E[\|Y_0'-Y_0\|^2_2] \\
    &\le 2M \E[\|Y_t'-Y_t\|^2_2]^{\frac{1}{2}} \E[\|Y_0' - Y_0\|^2_2]^{\frac{1}{2}}
    + 2\alpha \E[\|Y_0'-Y_0\|^2_2] \\
    &\le (2M + 2\alpha) \, \E[\|Y_0'-Y_0\|^2_2] \\
    &\le 4M \, \E[\|Y_0'-Y_0\|^2_2].
\end{align*}
Integrating, we conclude that for all $t \ge 0$,
\begin{align*}
    \E[\|(Y_t' - Y_0') - (Y_t - Y_0)\|^2_2] \le 4M \, \E[\|Y_0'-Y_0\|^2_2] \, t.
\end{align*}
\end{proof}

\subsection{Proof of Lemma~\ref{Lem:Growth}: Growth bound of MLD}

\begin{proof}[Proof of Lemma~\ref{Lem:Growth}]
Assume {\bf (A1)},  {\bf (A2)}, and {\bf (A3)}.
Consider the solution $Y_t$ of MLD~\eqref{Eq:MLD} starting from $Y_0$.
The centered variable $\tilde Y_t = Y_t - Y_0$ follows the SDE
\begin{align*}
    d\tilde Y_t = -g(\tilde Y_t + Y_0) dt + \sqrt{2} A(\tilde Y_t + Y_0) dW_t.
\end{align*}
Then
\begin{align*}
    \frac{d}{dt} \E[\|Y_t-Y_0\|^2_2] = \frac{d}{dt} \E[\|\tilde Y_t\|^2_2]
    &= -2\E[\langle g(\tilde Y_t + Y_0), \tilde Y_t \rangle] + 2\E[\|A(\tilde Y_t + Y_0)\|^2_{\HS}] \\
    &= \underbrace{-2\E[\langle g(Y_t), Y_t-Y_0 \rangle]}_{= ~ I} + \underbrace{2\E[\|A(Y_t)\|^2_{\HS}]}_{= ~ II}.
\end{align*}
Let us bound the two terms above.
Let $x^* = \arg\min_{x \in \X} f(x)$ 
and $y^* = \nabla \phi(x^*)$.

\paragraph{First term:}
By {\bf (A2)} and {\bf (A3)},
\begin{align*}
    I &= -2\E[\langle g(Y_t), Y_t-Y_0 \rangle] \\
    &= -2\E[\langle g(Y_t) - g(Y_0), Y_t-Y_0 \rangle] -2\E[\langle g(Y_0), Y_t-Y_0 \rangle] \\
    &\le -2\E[\langle g(Y_0), Y_t-Y_0 \rangle] \\
    &\le 2\E[\| g(Y_0) \| \cdot \|Y_t-Y_0\|] \\
    &\le \frac{1}{M^2} \E[\|g(Y_0)\|^2] + M^2 \E[\|Y_t-Y_0\|^2] \\
    &\le 2 \E[\|Y_0-y^\ast\|^2] + M^2 \E[\|Y_t-Y_0\|^2] + \frac{2}{M^2} \|g(y^\ast)\|^2.
\end{align*}
In the last step we have used $\|g(y)\|^2_2 \le 2\|g(y) - g(y^\ast)\|^2 + 2\|g(y^\ast)\|^2 \le 2M^2 \|y-y^\ast\|^2 + 2\|g(y^\ast)\|^2$.

\paragraph{Second term:}
By triangle inequality and {\bf (A1)},
\begin{align*}
    \|A(Y_t)\|^2_{\HS} 
    &\le 2\|A(Y_t) - A(Y_0)\|^2_{\HS} + 2\|A(Y_0)\|^2_{\HS} \\
    &\le 2\|A(Y_t) - A(Y_0)\|^2_{\HS} + 4\|A(Y_0)- A(y^\ast)\|^2_{\HS} + 4\|A(y^\ast)\|^2_{\HS} \\
    &\le 2\alpha \|Y_t-Y_0\|^2_2 + 4\alpha \|Y_0-y^\ast\|^2 + 4\|A(y^\ast)\|^2_{\HS}.
\end{align*}
Therefore,
\begin{align*}
    II &= 2\E[\|A(Y_t)\|^2_{\HS}] \\
    &\le 4\alpha \E[\|Y_t-Y_0\|^2_2] + 8\alpha \E[\|Y_0-y^\ast\|^2] + 8\|A(y^\ast)\|^2_{\HS}.
\end{align*}

Combining the two terms above, we get that along MLD~\eqref{Eq:MLD}:
\begin{align}\label{Eq:Calc1}
    \frac{d}{dt} \E[\|Y_t-Y_0\|^2_2] 
    &\le (M^2+4\alpha) \E[\|Y_t-Y_0\|^2_2] + D
\end{align}
where
\begin{align*}
D &= (2 + 8\alpha) \E[\|Y_0-y^\ast\|^2] + 8\|A(y^\ast)\|^2_{\HS} + \frac{2}{M^2} \|g(y^\ast)\|^2 \\
&\le 4(1 + 4\alpha) \E[\|Y_0\|^2] + 4(1+4\alpha) \|y^\ast\|^2 + 8\|A(y^\ast)\|^2_{\HS} + \frac{2}{M^2} \|g(y^\ast)\|^2.
\end{align*}
Recall in general if $V_t \ge 0$ satisfies $\frac{d}{dt} V_t \le C V_t + D$ for some $C,D > 0$, then $V_t \le e^{Ct} V_0 + \frac{D}{C}(e^{Ct}-1)$.
Furthermore, if $V_0 = 0$ and $0 < t \le \frac{1}{C}$, then 
    $V_t \le \frac{D}{C} 2Ct = 2Dt$.
Applying this to $V_t = \E[\|Y_t - Y_0\|^2]$ which satisfies~\eqref{Eq:Calc1} and $V_0 = 0$, we conclude that if $0 < t \le \frac{1}{M^2+4\alpha}$, then
\begin{align*}
    \E[\|Y_t-Y_0\|^2_2] \le \gamma t
\end{align*}
where $\gamma = 2D \le 8(1 + 4\alpha) \E[\|Y_0\|^2] + 8(1+4\alpha) \|y^\ast\|^2 + 16\|A(y^\ast)\|^2_{\HS} + \frac{4}{M^2} \|g(y^\ast)\|^2$.
\end{proof}

\subsection{Proof of Lemma~\ref{Lem:WeakError}: Local weak error of MLA}

\begin{proof}[Proof of Lemma~\ref{Lem:WeakError}]
Assume {\bf (A1)},  {\bf (A2)}, and {\bf (A3)}.
Starting from $Y_0 \in \R^d$, let $Y_t$ be the solution to the MLD~\eqref{Eq:MLD}, and let $Y_t'$ be the solution to the modified SDE with constant drift, driven by the same Brownian motion:
\begin{align*}
    dY_t &= -g(Y_t) dt + \sqrt{2} A(Y_t) dW_t \\
    dY_t' &= -g(Y_0) dt + \sqrt{2} A(Y_0) dW_t.
\end{align*}
The value $Y_h'$ at time $t = h$ is the output $\bar Y_1$ of MLA~\eqref{Eq:MLA2} from $Y_0$.
We wish to bound $\|\E[Y_h - Y_h']\|$.

Since $Y_t, Y_t'$ are coupled using the same Brownian motion, the difference $Y_t - Y_t'$ satisfies
\begin{align*}
    d(Y_t-Y_t') &= -(g(Y_t) - g(Y_0)) dt + \sqrt{2} (A(Y_t) - A(Y_0)) dW_t.
\end{align*}
Integrating, and since $Y_0 = Y_0'$, this means
\begin{align*}
    Y_h - Y_h' = -\int_0^h (g(Y_t) - g(Y_0)) dt + \sqrt{2} \int_0^h (A(Y_t) - A(Y_0)) dW_t.
\end{align*}
Taking expectation gives
\begin{align}\label{Eq:CalcWeak}
    \E[Y_h - Y_h'] = -\int_0^h \E[g(Y_t) - g(Y_0)] dt.
\end{align}
By {\bf (A2)} and Lemma~\ref{Lem:Growth}, for $0 < t \le \frac{1}{M^2+4\alpha}$ we have
\begin{align*}
\E[\|g(Y_t) - g(Y_0)\|] \le M \E[\|Y_t-Y_0\|] \le M \sqrt{\E[\|Y_t-Y_0\|^2]}  \le M \sqrt{\gamma t}.
\end{align*}
Therefore, by triangle inequality on~\eqref{Eq:CalcWeak}, for $0 < h \le \frac{1}{M^2+4\alpha}$ we have
\begin{align*}
    \|\E[Y_h - Y_h']\|
    &\le \int_0^h \E[\|g(Y_t) - g(Y_0)\|] dt \\
    &\le M \sqrt{\gamma} \int_0^h \sqrt{t} \, dt \\
    &= \frac{2}{3} M \sqrt{\gamma} \,h^{\frac{3}{2}} \\
    &= \frac{2}{3} M \left(8(1 + 4\alpha) \E[\|Y_0\|^2] + 8(1+4\alpha) \|y^\ast\|^2 + 16\|A(y^\ast)\|^2_{\HS} + \frac{4}{M^2} \|g(y^\ast)\|^2\right)^{\frac{1}{2}} h^{\frac{3}{2}} \\
    &\le \frac{2}{3} M \left( \sqrt{8(1 + 4\alpha)} \sqrt{\E[\|Y_0\|^2]} + \sqrt{8(1+4\alpha)} \|y^\ast\| + 4 \|A(y^\ast)\|_{\HS} + \frac{2}{M} \|g(y^\ast)\| \right) h^{\frac{3}{2}} \\
    &= \left(C_1 + D_1 \sqrt{\E[\|Y_0\|^2]} \right) h^{3/2}.
\end{align*}
This shows the local weak error order is at least $p_1 = \frac{3}{2}$, with maximum step size $h_1 = \frac{1}{M^2+8\alpha}$ and constants
\begin{align*}
    C_1 &= \frac{2}{3} M \left( \sqrt{8(1+4\alpha)} \|y^\ast\| + 4 \|A(y^\ast)\|_{\HS} + \frac{2}{M} \|g(y^\ast)\| \right) \\
    &\le 3M \sqrt{(1 + 4\alpha)} \,  \left( \|y^\ast\| + \|A(y^\ast)\|_{\HS} + \frac{1}{M} \|g(y^\ast)\| \right) \\
    D_1 &= \frac{2}{3} M \sqrt{8(1 + 4\alpha)} \\
    &\le 2M \sqrt{1 + 4\alpha}.
\end{align*}
\end{proof}

\subsection{Proof of Lemma~\ref{Lem:StrongError}: Local strong error of MLA}

\begin{proof}[Proof of Lemma~\ref{Lem:StrongError}] 

Assume~{\bf (A1)}, {\bf (A2)}, and {\bf (A3)}.
As in the proof of Lemma~\ref{Lem:WeakError}, consider two dynamics $Y_t', Y_t$ starting from $Y_0' = Y_0$ following the SDEs coupled with the same Brownian motion:
\begin{align*}
    dY_t &= -g(Y_t) dt + \sqrt{2} A(Y_t) dW_t \\
    dY_t' &= -g(Y_0) dt + \sqrt{2} A(Y_0) dW_t.
\end{align*}
We wish to bound $\E[\|Y_h - Y_h'\|^2]$.
The difference $Y_t - Y_t'$ satisfies
\begin{align*}
    d(Y_t-Y_t') &= -(g(Y_t) - g(Y_0)) dt + \sqrt{2} (A(Y_t) - A(Y_0)) dW_t.
\end{align*}
By~{\bf (A1)}, {\bf (A2)}, and Lemma~\ref{Lem:Growth}, for $0 < t \le \frac{1}{M^2+4\alpha}$ we have
\begin{align*}
    \frac{d}{dt} \E[\|Y_t-Y_t'\|^2_2] 
    &= -2\E[\langle g(Y_t) - g(Y_0),  Y_t-Y_t' \rangle] + 2 \E[\|A(Y_t) - A(Y_0)\|^2_{\HS}] \\
    &\le 2\E[\| g(Y_t) - g(Y_0) \|^2]^{\frac{1}{2}} \, \E[\|Y_t-Y_t'\|^2]^{\frac{1}{2}} + 2\alpha \E[\|Y_t - Y_0\|^2] \\
    &\le 2M\E[\| Y_t - Y_0 \|^2]^{\frac{1}{2}} \, \E[\|Y_t-Y_t'\|^2]^{\frac{1}{2}} + 2\alpha \E[\|Y_t - Y_0\|^2] \\
    &\le M^2 \E[\|Y_t-Y_t'\|^2] + (1+2\alpha) \E[\|Y_t - Y_0\|^2] \\
    &\le M^2 \E[\|Y_t-Y_t'\|^2] + (1 + 2\alpha) \gamma t.
\end{align*}
Equivalently, $\frac{d}{dt} (e^{-M^2 t} \E[\|Y_t-Y_t'\|^2_2]) \le e^{-M^2 t} (1 + 2\alpha) \gamma t \le (1 + 2\alpha) \gamma t$, so
\begin{align*}
    \E[\|Y_t-Y_t'\|^2_2] \le e^{M^2 t} \frac{(1 + 2\alpha)}{2} \gamma t^2.
\end{align*}
Furthermore, since $t \le \frac{1}{M^2+4\alpha} \le \frac{1}{M^2}$, we have $e^{M^2 t} \le e < 3$, so
\begin{align*}
    \E[\|Y_t-Y_t'\|^2_2] &\le \frac{3}{2} (1 + 2\alpha) \gamma t^2 \\
    &= 3 (1 + 2\alpha) \left(8(1 + 4\alpha) \E[\|Y_0\|^2] + 8(1+4\alpha) \|y^\ast\|^2 + 16\|A(y^\ast)\|^2_{\HS} + \frac{4}{M^2} \|g(y^\ast)\|^2\right) t^2 \\
    &= (C_2^2 + D_2^2 \, \E[\|Y_0\|^2]) \, t^2.
\end{align*}
This shows the local strong error order is at least $p_2 = 1$ with maximum step size $h_2 = \frac{1}{M^2+4\alpha}$ and constants
\begin{align*}
    C_2 &= \left(24 (1 + 2\alpha)(1 + 4\alpha) \|y^\ast\|^2 + 48 (1 + 2\alpha)\|A(y^\ast)\|^2_{\HS} + \frac{12(1+2\alpha)}{M^2} \|g(y^\ast)\|^2  \right)^{\frac{1}{2}} \\
           &\le 5 (1+4\alpha) \|y^\ast\| + 7\sqrt{1+2\alpha} \|A(y^\ast)\|_{\HS} + \frac{4\sqrt{1+2\alpha}}{M} \|g(y^\ast)\| \\
           &\le 7(1+4\alpha) \left( \|y^\ast\| + \|A(y^\ast)\|_{\HS} + \frac{1}{M} \|g(y^\ast)\| \right) \\
    D_2 &= \sqrt{24 (1 + 2\alpha)(1 + 4\alpha)} \\
    &\le 5 (1+4\alpha).
\end{align*}
\end{proof}

\section{An Analogy: Geometric Brownian Motion}
\label{App:GBM}

We wondered if our requirement on the modified self-concordance parameter $\alpha$ being upper-bounded is an artifact of our proof technique. Thus we did some simple calculations on {\bf Geometric Brownian Motion (GBM)} which is an SDE with multiplicative noise and yet admitting close-form solution. It is not an exact example of MLD but only an analogy; nevertheless, GBM does need $\alpha$ to be bounded in order to converge.

More precisely, consider GBM on $\R_+ = (0,\infty)$ which follows the stochastic differential equation:
\begin{align}\label{Eq:GBM}
dY_t = -Y_t \, dt + \sqrt{2\alpha} \, Y_t \, dW_t
\end{align}
where $dW_t$ is the standard Brownian motion on $\R$.
This has exact solution
\begin{align*}
Y_t = Y_0 \exp\left(-(1+\alpha)t + \sqrt{2\alpha} \, W_t\right).
\end{align*}

By a standard calculation, we see there is a threshold $\alpha < 1$ for the convergence of $Y_t$ as $t \to \infty$.
Recall since $W_t \sim \N(0,t)$, $\E[\exp(\sigma W_t)] = e^{\sigma^2 t/2}$ for all $\sigma > 0$.
Then
\begin{align*}
\E[Y_t^2] &= \E[Y_0^2] \, e^{-2(1+\alpha)t} \, \E[\exp(2\sqrt{2\alpha} W_t)]
  = \E[Y_0^2] \, e^{-2(1-\alpha)t}.  
\end{align*}
Therefore,
\begin{align*}
\lim_{t \to \infty} \E[Y_t^2] = \begin{cases}
0 ~~~ & \text{ if } \alpha < 1 \\
\E[Y_0^2] & \text{ if } \alpha = 1 \\
\infty & \text{ if } \alpha > 1.
\end{cases}
\end{align*}

Now consider a synchronous coupling $Y_t, \tilde Y_t$ following GBM~\eqref{Eq:GBM} with the same Brownian motion:
\begin{align*}
Y_t &= Y_0 \exp\left(-(1+\alpha)t + \sqrt{2\alpha} \, W_t\right) \\
\tilde Y_t &= \tilde Y_0 \exp\left(-(1+\alpha)t + \sqrt{2\alpha} \, W_t\right).
\end{align*}
Then
\begin{align*}
\E[(Y_t - \tilde Y_t)^2] &= \E[(Y_0 - \tilde Y_0)^2] \, e^{-2(1+\alpha)t} \, \E[\exp(2\sqrt{2\alpha} W_t)]
= \E[(Y_0 - \tilde Y_0)^2] \, e^{-2(1-\alpha)t}.
\end{align*}
Thus, we see that GBM is a contraction if and only if $\alpha < 1$.
In particular, we also have
\begin{align*}
\lim_{t \to \infty} \E[(Y_t-\tilde Y_t)^2] = \begin{cases}
0 ~~~ & \text{ if } \alpha < 1 \\
\E[(Y_0-\tilde Y_0)^2] & \text{ if } \alpha = 1 \\
\infty & \text{ if } \alpha > 1.
\end{cases}
\end{align*}

GBM~\eqref{Eq:GBM} is an instance of MLD~\eqref{Eq:MLD} (and in fact NLD~\eqref{Eq:NLD}) with $\phi = f$ where
\begin{align}\label{Eq:GBMf}
\frac{1}{\sqrt{(\phi^\ast)''(y)}} = \sqrt{\alpha} y
\end{align}
which is $\sqrt{\alpha}$-Lipschitz, so it satisfies modified self-concordance~{\bf (A1)} with parameter $\alpha$.
Since $\phi = f$, it satisfes relative smoothness~{\bf (A2)} and relative strong convexity~{\bf (A3)} with $M = m = 1$.
Note our assumption in Theorem~\ref{Thm:Main} is $\alpha < m = 1$, which  c is tight for GBM to contract, as well as to determine if there is a $t\to\infty$ limit.

\section{Example: Log-Barrier on a Polytope}
\label{App:LogBarrier}

Let $\X$ be the polytope (not necessarily bounded)
$$\X = \{x \in \R^d \colon a_i^\top x \ge b_i ~~ \forall \, i = 1,\dots,m\}$$
for some $a_1,\dots,a_m \in \R^d$ and $b_1,\dots,b_m \in \R$.
Consider the log-barrier function defined in the interior of $\X$:
\begin{align}
    \phi(x) = -\sum_{i=1}^m \log (a_i^\top x - b_i).
\end{align}
Recall that $\phi$ satisfies the classical self-concordance condition with a constant parameter $2$.
Let $\alpha$ be the modified self-concordance parameter of $\phi$.
For some polytopes, such as the positive orthant, $\alpha$ is also a constant (because the Hessian is diagonal and the dimensions are independent).
For general polytopes, however, $\alpha$ can be arbitrarily large.
Here we show $\alpha$ can be as large as the square inverse of the smallest singular value of the constraint matrix; we also construct an explicit example in two dimension.

Without loss of generality we may assume $\|a_i\| = 1$ for $i = 1,\dots,m$.
Let $A = ( a_1, \cdots, a_m ) \in \R^{d \times m}$ be the constraint matrix, so the polytope is described by $A^\top x \ge b$.
Let the singular values of $A$ be $\sigma_1 \ge \dots \ge \sigma_d \ge 0$ (assuming $d \le m$).
Then $\sum_{i=1}^d \sigma_i^2 = \Tr(AA^\top) = \Tr(A^\top A) = \sum_{i=1}^m \|a_i\|^2 = m$;
but $\sigma_d = \min_i \sigma_i$ can be small or $0$.
For $x \in \X$, let $S_x \in \R^{m \times m}$ be the diagonal matrix with entries $a_i^\top x - b_i$.

The gradient of $\phi$ is
\begin{align*}
    \nabla \phi(x) = -\sum_{i=1}^m \frac{a_i}{a_i^\top x - b_i}.
\end{align*}
Then for $x, x' \in \X$, we have
\begin{align*}
    \nabla \phi(x') - \nabla \phi(x)
    &= \sum_{i=1}^m \left(\frac{1}{a_i^\top x - b_i} - \frac{1}{a_i^\top x' - b_i} \right) a_i \\
    &= \sum_{i=1}^m \frac{a_i^\top (x'-x)}{(a_i^\top x - b_i)(a_i^\top x' - b_i)} a_i \\
    &= \sum_{i=1}^m \frac{a_i a_i^\top}{(a_i^\top x - b_i)(a_i^\top x' - b_i)} (x'-x) \\
    &= A S_x^{-1} S_{x'}^{-1} A^\top (x'-x).
\end{align*}
Therefore,
\begin{align*}
    \|\nabla \phi(x') - \nabla \phi(x)\|^2 
    &= \|A S_x^{-1} S_{x'}^{-1} A^\top (x'-x) \|^2 \\
    &= (x'-x)^\top A S_{x'}^{-1} S_x^{-1} A^\top A S_x^{-1} S_{x'}^{-1} A^\top (x'-x) \\
    &= v(x,x')^\top A^\top A v(x,x')
\end{align*}
where 
$$v(x,x') = S_x^{-1} S_{x'}^{-1} A^\top (x'-x) \in \R^m.$$

The Hessian is
\begin{align*}
    \nabla^2 \phi(x) = \sum_{i=1}^m \frac{a_i a_i^\top}{(a_i^\top x - b_i)^2} = A S_x^{-2} A^\top.
\end{align*}
As a square-root, we can choose:
\begin{align*}
    \sqrt{\nabla^2 \phi(x)} = A S_x^{-1}
    = \begin{pmatrix} \dfrac{a_1}{a_1^\top x - b_1} & \cdots & \dfrac{a_m}{a_m^\top x - b_m}
    \end{pmatrix}
\end{align*}
since indeed $\sqrt{\nabla^2 \phi(x)} \sqrt{\nabla^2 \phi(x)}^\top = A S_x^{-1} S_x^{-1} A^\top = \nabla^2 \phi(x)$.

For $x,x' \in \X$, we have that 
\begin{align*}
    \sqrt{\nabla^2 \phi(x')} - \sqrt{\nabla^2 \phi(x)}
    &= A (S_{x'} - S_x) \\
    &= -\begin{pmatrix} \dfrac{a_1 a_1^\top (x'-x)}{(a_1^\top x' - b_1)(a_1^\top x - b_1)} & \cdots & \dfrac{a_m a_m^\top (x'-x)}{(a_m^\top x' - b_m)(a_m^\top x - b_m)}
    \end{pmatrix}.
\end{align*}
Therefore,
\begin{align*}
    \|\sqrt{\nabla^2 \phi(x')} - \sqrt{\nabla^2 \phi(x)}\|^2_{\HS} 
    &= \sum_{i=1}^m \left\| \frac{a_1 a_1^\top (x'-x)}{(a_1^\top x' - b_1)(a_1^\top x - b_1)} \right\|^2 \\
    &= \sum_{i=1}^m (x'-x)^\top \frac{a_1 a_1^\top a_1 a_1^\top }{(a_1^\top x' - b_1)^2(a_1^\top x - b_1)^2} (x'-x) \\ 
    &= (x'-x)^\top \left( \sum_{i=1}^m \frac{a_1 a_1^\top }{(a_1^\top x' - b_1)^2(a_1^\top x - b_1)^2} \right) (x'-x) \\ 
    &= (x'-x)^\top AS_{x'}^{-2} S_x^{-2} A^\top (x'-x) \\
    &= \|v(x,x')\|^2.
\end{align*}

\paragraph{Modified self-concordance.}
The modified self-concordance parameter is
\begin{align*}
    \alpha &= \sup_{x,x' \in \X} \frac{\|\sqrt{\nabla^2 \phi(x')} - \sqrt{\nabla^2 \phi(x)}\|^2_{\HS}}{\|\nabla \phi(x') - \nabla \phi(x)\|^2_2} \\
    &= \sup_{x,x' \in \X} \frac{\|v(x,x')\|^2}{v(x,x')^\top  A^\top A v(x,x')} \\
    &\le \sup_{v \in \R^d} \frac{\|v\|^2}{v^\top A^\top A v} \\
    &= \max_{i = 1,\dots,d} \frac{1}{\sigma_i^2} \\
    &= \frac{1}{\sigma_d^2}.
\end{align*}
This shows the modified self-concordance parameter can be as large as $\frac{1}{\sigma_d^2}$, by choosing appropriate $x,x'$.
For some polyhedra $\sigma_d \approx 0$, so $\alpha \approx 1/\sigma_d^2$ can be arbitrarily large.

\paragraph{Example in two dimension.}
Let $d=2$, and consider
$$a_1 = 
\begin{pmatrix}
1 \\ 0
\end{pmatrix},
~~~~~~
a_2 = 
\begin{pmatrix}
\sqrt{1-\epsilon^2} \\ \epsilon
\end{pmatrix}
$$
for some small $\epsilon > 0$, and
$b_1 = b_2 = 0$.
This defines the intersection of two halfspaces:
$$\X = \{x = (x_1,x_2) \colon x_1 \ge 0, ~  \sqrt{1-\epsilon^2} \,  x_1 + \epsilon x_2 \ge 0\}.$$
The constraint matrix is
$A = \begin{pmatrix}
1 & \sqrt{1-\epsilon^2} \\
0 & \epsilon
\end{pmatrix}.$
We have 
$A^\top A = \begin{pmatrix}
1 & \sqrt{1-\epsilon^2} \\
\sqrt{1-\epsilon^2} & 1
\end{pmatrix}$
which has eigenvalues $\sigma_1^2 = 1+\sqrt{1-\epsilon^2}$ and $\sigma_2^2 = 1-\sqrt{1-\epsilon^2}$.
Note that if $\epsilon$ is small, $\sigma_1^2 \approx 2$ and $\sigma_2^2 \approx \epsilon^2/2$.
The corresponding eigenvectors are $v_1 = \begin{pmatrix}
1 \\ 1
\end{pmatrix}$ and $v_2 = \begin{pmatrix}
1 \\ -1
\end{pmatrix}$.

Let us choose
$$x = \begin{pmatrix}
1 \\ 0
\end{pmatrix}, ~~~~~
x' = \begin{pmatrix}
a \\ b
\end{pmatrix}$$
for some constant $a,b \in \R$.
For simplicity let $s = \sqrt{1-\epsilon^2}$. 
We require $x' \in \X$, so $a \ge 0$ and $b \ge -\frac{s}{\epsilon}  a$.
We have 
$$A^\top x = \begin{pmatrix}
1 \\ s
\end{pmatrix}, ~~~~~~
A^\top x' = \begin{pmatrix}
a \\ s a + \epsilon b
\end{pmatrix}
$$
and
$$A^\top(x'-x) = \begin{pmatrix}
a-1 \\ s (a-1) + \epsilon b
\end{pmatrix}.
$$
We also have
$$S_x = \begin{pmatrix}
1 & 0 \\
0 & s
\end{pmatrix}, ~~~~~
S_x = \begin{pmatrix}
a & 0 \\
0 & s a + \epsilon b
\end{pmatrix}.
$$
Then
\begin{align*}
v(x,x') 
&= S_x^{-1} S_{x'}^{-1} A^\top (x'-x) \\
&= \begin{pmatrix}
\frac{1}{a} & 0 \\
0 & \frac{1}{s (s a + \epsilon b)}
\end{pmatrix}
\begin{pmatrix}
a-1 \\ s(a-1) + \epsilon b
\end{pmatrix} \\
&= \begin{pmatrix}
\frac{a-1}{a} \\ \frac{s(a-1) + \epsilon b}{s ( sa + \epsilon b)}
\end{pmatrix}
\end{align*}
We want this to be proportional to $v_2 = \begin{pmatrix}
1 \\ -1
\end{pmatrix}$, 
so we want
\begin{align}\label{Eq:Want}
    \frac{a-1}{a} + \frac{s(a-1) + \epsilon b}{s (sa + \epsilon b)} = 0.
\end{align}
We can solve for $b$ in terms of $a$:
$$b = -\frac{a(a-1) s (s+1)}{\epsilon((a-1) s + a)}.$$
We can see that for all $a \ge 0$, this choice of $b$ satisfies the constraint $b \ge -\frac{s}{\epsilon}  a$, so $x' \in \X$.

Explicitly, suppose we choose
\begin{align*}
    a &= 2 \\
    b &= -\frac{2 s(s+1)}{\epsilon(s+2)}
\end{align*}
which satisfies the condition $b \ge -\frac{2s}{\epsilon}$.
We can verify directly that the condition~\eqref{Eq:Want} holds:
\begin{align*}
    \frac{a-1}{a} + \frac{s(a-1) + \epsilon b}{s ( s a + \epsilon b)}
    &= \frac{1}{2} + \frac{s -\frac{2 s(s+1)}{(s+2)}}{s ( 2s -\frac{2 s(s+1)}{(s+2)})} 
    = \frac{1}{2} + \frac{-\frac{s^2}{s+2}}{s(\frac{2s}{(s+2)})} 
    = \frac{1}{2} - \frac{1}{2} 
    = 0.
\end{align*}
Then with this choice
$$x = \begin{pmatrix}
1 \\ 0
\end{pmatrix}, ~~~~~
x' = \begin{pmatrix}
2 \\ -\frac{2 s(s+1)}{\epsilon(s+2)}
\end{pmatrix}$$
we have that $v(x,x') = \frac{1}{2} v_2$, i.e.\ proportional to the eigenvector of $A^\top A$ with small eigenvalue $\sigma_2^2$.
Then $A^\top A v(x,x') = \sigma_2^2 v(x,x')$, and this gives the bound for the modified self-concordance parameter:
\begin{align*}
\alpha \ge \frac{\|v(x,x')\|^2}{v(x,x')^\top A^\top A v(x,x')}
= \frac{\|v(x,x')\|^2}{\sigma_2^2 \|v(x,x')\|^2} = \frac{1}{\sigma_2^2}
= \frac{1}{1-\sqrt{1-\epsilon^2}}
\approx \frac{2}{\epsilon^2}.
\end{align*}
Thus, by setting $\epsilon \to 0$ we can make $\alpha$ as large as we want.
However, note that the case $\epsilon = 0$ is nice and we have $\alpha = 1$, because the domain is a half-space and the problem reduces to one dimension.
This example shows the definition of modified self-concordance is not stable.

\end{document}